\documentclass[copyright,creativecommons]{eptcs}

\usepackage{amssymb}
\usepackage{amsmath}
\usepackage{amsthm}
\usepackage{breakurl}
\usepackage[mathscr]{euscript}
\usepackage{graphicx}
\usepackage{mathtools}
\usepackage{times}
\usepackage{xcolor}

\DeclareMathAlphabet{\mathcal}{OMS}{cmsy}{m}{n}

\providecommand{\event}{QAPL 2015}
\providecommand{\firstpage}{1}
\providecommand\titlerunning{A Definition Scheme for Quantitative Bisimulation}
\providecommand\authorrunning{Latella, Massink \& De Vink}

\mathcode`:="103A       
\mathcode`;="203B       
\mathcode`?="003F       
\mathcode`|="026A       
\mathcode`<="4268       
\mathcode`>="5269       

\mathchardef\ls="213C   
\mathchardef\gr="213E   
\mathchardef\uparrow="0222      
\mathchardef\downarrow="0223    

\newcommand{\blankline}{\vspace*{\baselineskip}}

\newcommand{\nop}[1]{}

\newcommand{\CTMC}{\textsl{CTMC} \mkern1mu}
\newcommand{\DTMC}{\textsl{DTMC} \mkern1mu}
\newcommand{\Distr}{\textsl{Distr} \mkern1mu}
\newcommand{\FALSE}{\textbf{false}}
\newcommand{\FS}{\mkern1mu \mathrm{\mathcal{F} \mkern-2.5mu \mathcal{S}}}
\newcommand{\FuTS}{\textsl{Fu\hspace*{-0.5pt}TS} \mkern1mu}

\newcommand{\IMC}{\textsl{IMC} \mkern1mu}
\newcommand{\LTS}{\textsl{LTS} \mkern1mu}
\newcommand{\MA}{\textsl{MA} \mkern1mu}
\newcommand{\PA}{\textsl{PA} \mkern1mu}
\newcommand{\PEPA}{\textsl{PEPA} \mkern1mu}

\newcommand{\SPC}{\textsl{SPC} \mkern1mu}
\newcommand{\Set}{\textbf{Set}}
\newcommand{\TRUE}{\textbf{true}}
\newcommand{\ULTraS}{\textsl{ULTraS} \mkern1mu}

\newcommand{\approxT}{\approx_{\mkern2mu \calT}}
\newcommand{\bfOmega}{\boldsymbol{\Omega}}

\newcommand{\bfR}{\mathbf{R}}
\newcommand{\bfT}{\mathbf{T}}
\newcommand{\bools}{\mathbb{B}}
\newcommand{\compose}{\mathop{\raisebox{0.5pt}{\scriptsize $\circ$}}}
\newcommand{\fmorph}[3]{[\![ {#3} ]\!]^{#2}_{#1}}
\newcommand{\fsfn}[2]{\mkern1mu \mathrm{\mathcal{F} \mkern-2.5mu \mathcal{S}}( \mkern1mu #1,#2 \mkern2mu )}
\newcommand{\lift}{{\cal LT}}
\newcommand{\lc}{\mathopen{\lbrace \;}}
\newcommand{\mtrans}[2]{\, \xRightarrow{\, #1 \,}_{#2} \,}
\newcommand{\myell}{{\ell \mkern2mu}}

\newcommand{\nnreals}{\reals_{\geqslant 0}}

\newcommand{\rc}{\mathclose{\; \rbrace}}
\newcommand{\reals}{\mathbb{R}}
\newcommand{\singleton}[1]{\lbrace {#1} \rbrace}
\newcommand{\sbar}{\bar{s}}
\newcommand{\thetaC}{\theta_\calC}
\newcommand{\thetaE}{\theta_E}
\newcommand{\thetaI}{\theta_\calI}
\newcommand{\thetaII}{\theta^{\mkern1mu \mathrm{I}}_\calI}
\renewcommand{\thetaII}{\theta'_\calI}
\newcommand{\thetaIM}{\theta'_\calM}

\newcommand{\thetaL}{\theta_\calL}
\newcommand{\thetaM}{\theta_\calM}
\newcommand{\thetaMI}{\theta^{\mkern1mu \mathrm{M}}_\calI}
\renewcommand{\thetaMI}{\theta''_\calI}
\newcommand{\thetaMM}{\theta''_\calM}

\newcommand{\thetaP}{\theta_\calP}
\newcommand{\trans}[2]{\, \xrightarrow{\, #1 \,}_{#2} \,}
\newcommand{\tssum}{\textstyle{\sum \,}}
\newcommand{\zbar}{\bar{z}}

\newcommand{\calA}{\mathscr{A}}
\newcommand{\calC}{\mathscr{C}}

\newcommand{\calF}{\mathscr{F}}
\newcommand{\calI}{\mathscr{I}}
\newcommand{\calL}{\mathscr{L}}
\newcommand{\calM}{\mathscr{M}}
\newcommand{\calP}{\mathscr{P}}
\newcommand{\calR}{\mathscr{R}}

\newcommand{\calT}{\mathscr{T}}
\newcommand{\calX}{\mathscr{X}}

\newtheorem{theorem}{Theorem}
\newtheorem{lemma}{Lemma}
\newtheorem{definition}[theorem]{Definition}

\title{A Definition Scheme for Quantitative Bisimulation}
\author{%
  Diego Latella \& Mieke Massink
  \institute{Formal Methods and Tools \\ 
    CNR/ISTI, Pisa, Italy}
  \and
  Erik de Vink\thanks{Corresponding author, email~\url{evink@win.tue.nl}.}
  \institute{Department of Mathematics and Computer Science, TU/e \\
    Eindhoven, the Netherlands \\
    Centrum voor Wiskunde en Informatica \\ 
    Amsterdam, the Netherlands}
}

\begin{document}

\maketitle

\begin{abstract}
  \textbf{Abstract} $\FuTS$, state-to-function transition systems
  are generalizations of labeled transition systems and of familiar
  notions of quantitative semantical models as continuous-time Markov
  chains, interactive Markov chains, and Markov automata. A general
  scheme for the definition of a notion of strong bisimulation
  associated with a $\FuTS$ is proposed. It is shown that this notion
  of bisimulation for a $\FuTS$ coincides with the coalgebraic notion
  of behavioral equivalence associated to the functor on~$\Set$ given by
  the type of the~$\FuTS$. For a series of concrete quantitative
  semantical models the notion of bisimulation as reported in the
  literature is proven to coincide with the notion of quantitative
  bisimulation obtained from the scheme. The comparison includes
  models with orthogonal behaviour, like interactive Markov chains,
  and with multiple levels of behavior, like Markov automata. As a
  consequence of the general result relating $\FuTS$ bisimulation and
  behavioral equivalence we obtain, in a systematic way, a coalgebraic
  underpinning of all quantitative bisimulations discussed.

  \medskip

\end{abstract}


\section{Introduction}
\label{sec:intro}

State-to-Function Labeled Transition Systems ($\FuTS$) have been
introduced in~\cite{DLLM13:cs} as a general framework for the formal
definition of the semantics of process calculi, in particular of
stochastic process calculi ($\SPC$), as an alternative to the
classical approach, based on labelled transition
systems~($\LTS$). In~$\LTS$, a transition is a triple $(s, \alpha, s'
\mkern1mu )$ where $s$ and~$\alpha$ are the source state and the label
of the transition, respectively, while $s'$ is the target state
reached from~$s$ via a transition labeled with~$\alpha$. In~$\FuTS$, a
transition is a triple of the form $(s, \alpha, \varphi \mkern2mu
)$. The first and second component are the source state and the label
of the transition, as in~$\LTS$, while the third component~$\varphi$
is a \emph{continuation function} (or simply a \emph{continuation} in
the sequel), which associates a value from an appropriate semiring
with each and every state~$s'$. If~$\varphi$~maps $s'$ to the
$0$~element of the semiring, then state~$s'$ cannot be reached
from~$s$ via this transition. A non-zero value for a state~$s'$
represents a quantity associated with the jump of the system from $s$
to~$s'$. For instance, for continuous-time Markov chains ($\CTMC$),
the model underlying prominent~$\SPC$, this quantity is the rate of
the negative exponential distribution characterizing the time for the
execution of the action represented by~$\alpha$, necessary to reach
$s'$ from~$s$ via the transition. 

We note that for the coalgebraic treatment itself of $\FuTS$ we
propose here it is not necessary for the co-domain of continuations to
be semirings; working with monoids would be sufficient. However, the
richer structure of semirings is convenient, if not essential, when
using continuations and their operators in the formal definition of
the $\FuTS$ semantics of $\SPC$.  The use of continuations provides a
clean and simple solution to the transition multiplicity problem and
makes $\FuTS$ with~$\nnreals$ as semiring particularly suited for
$\SPC$~semantics. We refer to~\cite{DLLM13:cs} for a thorough
discussion on the use of~$\FuTS$ as the underlying semantics model for
$\SPC$---including those where (continuous-time) stochastic behaviour
is integrated with non-determinism as well as with discrete
probability distributions over behaviours---and for their comparison
with other recent approaches to a uniform treatment of
$\SPC$~semantics, e.g.\ Rated Transition Systems~\cite{KS08:fossacs},
Rate Transition Systems~\cite{DLLM09:icalp}, Weighted Transition
Systems~\cite{Kli09:mosses}, and $\ULTraS$~\cite{BDL13:ic}.



In~\cite{LMV12:accat} we proposed a coalgebraic view of~$\FuTS$
involving functors $\fsfn{{\cdot}}{\calR}$, for $\calR$ a semiring, in
place to deal with quantities. Here, for a set of states~$S$, we use
$\fsfn{S}{\calR}$ to denote the set of all finitely supported
functions from~$S$ to~$\calR$. Discrete probability distributions are
a typical example of such a function space. Note, the functor
$\fsfn{{\cdot}}{\calR}$ over a monoid or semiring~$\calR$, unifies the
finite subsets functor $\calP_\omega( {\cdot} )$ and the discrete
probability distributions functor $\Distr({\cdot})$, choosing the
semirings $\bools$ and~$\nnreals$, respectively (and requiring
restrictions on the continuations for the latter). Coalgebras, also
built from the finite support functor, are of interest because they
come equipped, under mild conditions, with a canonical notion of
`bisimulation' known as behavioral equivalence~\cite{Rut00:tcs}. Our
earlier work presented a notion of bisimulation for~$\FuTS$ and we
proved a correspondence result stating that $\FuTS$ bisimilarity
coincides with the behavioral equivalence of the associated
functor. We applied the framework to prominent~$\SPC$ like
$\PEPA$~\cite{Hil96:phd} and a language for Interactive Markov Chains
($\IMC$) of~\cite{Her02:springer}, thus providing coalgebraic
justification of the equivalences of these calculi.
In~\cite{MP14:qapl} an approach similar to ours has been applied to
the $\ULTraS$ model, a model which shares some features with simple
$\FuTS$. In~$\ULTraS$ posets are used instead of semirings, although a
monoidal structure is then implicitly assumed when process
equivalences are taken into consideration~\cite{BDL13:ic}.

An interesting direction of research combining coalgebra and
quantities investigates various types of weighted automata, including
linear weighted automata, and associated notions of bisimulation and
languages, as well as algorithms for these
notions~\cite{BBBRS12,Kli09:mosses}. Klin considers weighted
transition systems, labelled transition systems that assign a weight
to each transition and develops Weighted GSOS, a meta-syntactic
framework for defining well-behaved weighted transition systems. For
commutative monoids the notion of a weighted transition system
compares with our notion of a~$\FuTS$, and, when cast in the
coalgebraic setting, the associated concept of bisimulation coincides
with behavioral equivalence.  Weights of transitions of weighted
transition systems are computed by induction on the syntax of process
terms and by taking into account the contribution of all those GSOS
rules that are triggered by the relevant (apparent) weights. Note that
such a set of rules is finite.  So, in a sense, the computation of the
weights is distributed among (the instantiations of) the relevant
rules with intermediate results collected and integrated in the final
weight. In~\cite{MP14:qapl} a general GSOS specification format is
presented which allows for a `syntactic' treatment of continuations
involving so-called~\emph{weight functions}.  A comparison of a wide
range of probabilistic transition systems focusing on coalgebraic
bisimulation is reported in ~\cite{BSV04:tcs,Sok11:tcs}, which provide
hierarchy relating types of transition systems via natural
embeddings. In~\cite{LMV13} the investigation on the relationship
between $\FuTS$ bisimilarity and behavioural equivalence, and also
coalgebraic bisimilarity is presented. In particular, it is shown that
the functor type involved preserves weak pullbacks when the underlying
semiring satisfies the zero-sum property.

So far, all the unifying approaches to modeling of $\SPC$ discussed
above restrict to a single layer of quantities. Although for~$\IMC$
orthogonal transition relations, hence products of continuation sets
need to be considered, overall no nesting of the construction with
finitely supported functions is allowed. The problem appears to lie,
at least for~$\FuTS$, in identifying an appropriate notion of
bisimulation. Bisimulation for probabilistic systems is traditionally
based on equivalence classes \cite{LS91:ic,Bai98:hab,CR11:concur} with
a lifting operator from states to probability distributions as a main
ingredient. In the present paper, we extend the results reported
in~\cite{LMV12:accat} by taking this ingredient of lifting into
account systematically, thus also catering for repeated lifting. In
line with the coalgebraic paradigm, the type of the $\FuTS$, hence not
its transitions, decides the way equivalence classes need to be
lifted. The generalization paves the way to deal with more intricate
interactions of qualitative and quantitative behaviour.  For instance,
now we are capable to deal uniformly with Probabilistic Automata
($\PA$), see e.g.~\cite{SL95:njc,Hen12:facs}, and Markov Automata
($\MA$), cf.~\cite{EHZ10:lics,EHZ10:concur,TKPS12:concur}, which do
not fit in the unifying treatments mentioned above. Thus, the present
paper constitutes an additional step forward to implementing the aim
of a uniform treatment of semantic models for quantitative process
calculi and associated notion of strong bisimulation, but avoiding
implicit transition multiplicity and/or explicit transition decoration
(see~\cite{DLLM13:cs}). In addition, we have coalgebra interpretations
and bisimulation correspondence results at our disposal as a yardstick
for justifying particular process equivalence as being the natural
ones.

The paper is structured as follows: Section~\ref{sec:preliminaries}
introduces notation and concepts related to the nesting of constructs
of the form~$\fsfn{{\cdot}}{\calR}$, for a semiring~$\calR$, and
briefly discusses coalgebraic notions. Section~\ref{sec:futs} recalls
the notion of a~$\FuTS$, distinguishing simple, combined, nested and
general~$\FuTS$. The notion of bisimulation for a~$\FuTS$, the scheme
of defining quantitative bisimulation and the comparison of the latter
with behavioral equivalence is addressed too. Section~\ref{sec:simple}
presents the treatment of~$\LTS$, of~$\CTMC$ and of~$\IMC$
with~$\FuTS$. In particular, the concrete notions of strong
bisimulation for these semantical models are related to $\FuTS$
bisimulation. Section~\ref{sec:nested} continues this, now for $\PA$
and~$\MA$, semantical models that involve nesting. Also for these
models concrete bisimulation and $\FuTS$ bisimulation, as obtained
from the general scheme, are shown to
coincide. Section~\ref{sec:conclusions} wraps up with concluding
remarks.


\section{Preliminaries}
\label{sec:preliminaries}

A semiring~$\calR$ is tuple $( R , {+} , 0 , {\cdot} \mkern2mu , 1 )$
with $(R , {+} , 0 )$ a commutative monoid with zero element~$0$, $(
R, {\cdot} \mkern2mu , 1 )$ a monoid with identity element~$1$ such
that the multiplication~$\cdot$ distributes over the addition~$+$, and
the zero element~$0$ annihilates~$R$, i.e.\ $0 \cdot r = 0 = r \cdot
0$.  A function $f : X \to R$ from a set~$X$ to the carrier of a
semiring~$\calR$ is said to be of finite support if the set $\lc x \in
X \mid f(x) \neq 0 \rc$ is finite. We use $\fsfn{X}{\calR}$ to denote
the set of all functions from~$X$ to~$\calR$ of finite support.  The
shorthand $\fsfn{X}{\calR_1, \ldots, \calR_n}$, for a set~$X$ and
semirings $\calR_1, \ldots , \calR_n$, for $n \geqslant 0$, is given
by $\FS(X) = X$, and $\fsfn{X}{\calR_1, \ldots, \calR_n, \calR_{n+1}}
= \fsfn{ \, \fsfn{X}{\calR_1, \ldots, \calR_n} \,}{\calR_{n+1}}$.  We
let $\fsfn{X}{\calR}^{\calL}$ denote the set of functions from $\calL$
to $\fsfn{X}{\calR}$ and we extend the notation to $\fsfn{X}{\calR_1,
  \ldots, \calR_n}^{\calL}$ in the obvious way.


For a set~$X$, a semiring~$\calR$, a set of finitely supported
functions $F \subseteq \fsfn{X}{\calR}$, and an equivalence
relation~$E$ on~$X$, we define the relation $\lift(E,\calR)_F$ with respect
to~$F$ by $\lift(E,\calR)_F(\varphi,\psi) \iff \varphi[B] = \psi[B]$ for all
$B \in X/E$, for $\varphi,\psi \in F$. Here $\chi[A]$ denotes $\sum_{a
  \in A} \: \chi(a)$, for $\chi \in \fsfn{X}{\calR}$ and $A \subseteq
X$. Note that $\lift(E,\calR)_F$ is an equivalence relation on~$F$. We use
$\lift(E,\calR)$ to denote the lifting of~$E$ with respect to
$\fsfn{X}{\calR}$ itself.
For a sequence of semirings $\calR_1 , \ldots , \calR_n$, we define
the relation $\lift(E, \calR_1, \ldots,\calR_n)$
on~$\fsfn{X}{\calR_1, \ldots , \calR_n}$ by 
$\lift(E) = E$, and 
$\lift(E,\calR_1, \ldots , \calR_{n+1})(\varphi,\psi) \iff
\varphi[C] = \psi[C]$ for all
$C \in \fsfn{X, \calR_1, \ldots}{\calR_n}/\lift(E,\calR_1, \ldots , \calR_{n})$ and
$\varphi,\psi \in \fsfn{X, \calR_1, \ldots}{\calR_{n+1}}$.
%
 By induction on~$n$ one 
establishes that $\lift(E,\calR_1, \ldots , \calR_{n})$ is an
equivalence relation too.

If $E$ and~$F$ are binary relations of the sets
$X$ and~$Y$, respectively, we define the relation $E \times F
\subseteq (X \times Y) \times (X \times Y)$ by $E \times F = \lc (
<x_1,y_1> , <x_2,y_2> ) \mid E(x_1,x_2) \land F(y_1,y_2) \rc$. It holds
that $E \times F$ is an equivalence relation if $E$ and~$F$ are.

For a commutative monoid, hence for a semiring or a field, the functor
$\fsfn{{\cdot}}{\calR} : \Set \to \Set$, on the category~$\Set$ of sets
and functions, assigns to a set~$X$ the set of all finitely supported
functions~$\fsfn{X}{\calR}$, and to a function $f : X \to Y$ the
function $\fsfn{f}{\calR} : \fsfn{X}{\calR} \to \fsfn{Y}{\calR}$ given
by 
\begin{displaymath}
  \fsfn{f}{\calR}( \varphi )( y )
  =
  \tssum_{x \in f^{-1}(y)} \: \varphi(x)
\end{displaymath}
Note, since $\varphi$ is assumed to have a finite support, the
summation at the right side in~$\calR$ may be infinite, but is still
well-defined because of commutativity of~$+$ on~$\calR$. Also note,
only addition of~$\calR$ is used above. However, in concrete
situations, in particular when modeling the parallel operator
for~$\SPC$, multiplication of~$\calR$ is needed as well,
cf.~\cite{LMV12:accat,LMV13}. 

A pair of a set and a mapping $(X,\alpha)$ is called a coalgebra of a
functor~$F$ on~$\Set$ if $\alpha : X \to FX$. A mapping $f : X \to Y$
is a coalgebra homomorphism for two coalgebras $(X,\alpha)$
and~$(Y,\beta)$, if $Ff \compose \alpha = \beta \compose f$. A final
coalgebra of~$F$ is a coalgebra $( \Omega , \omega )$ of~$F$ such that
for every coalgebra~$X$ there exists a unique coalgebra
homomorphism~$\fmorph{F}{X}{{\cdot}} : X \to \Omega$. A $\Set$
functor~$F$ is accessible, if it preserves $\kappa$-filtered colimits,
for some regular cardinal~$\kappa$. Typically, one uses the following
characterization of accessibility: every element of~$FX$, for any
set~$X$, lies in the image of some subset~$Y \subseteq X$ of less than
$\kappa$~elements~\cite{AP04:tcs}. In essence, familiar functors like
the finite powerset functor~$\calP_\omega( {\cdot} )$ and the discrete
probability distribution functor~$\Distr( {\cdot} )$ are
accessible. Also, their generalization used in this paper, viz.\ the
finitely supported functions functor
$\fsfn{{\cdot}}{\calR}$~\cite{Kli09:mosses,BBBRS12} on the
category~$\Set$ is accessible as well. Accessibility is preserved by
exponentiation and products. Therefore, all functors appearing in the
sequel are accessible.

The appealing feature of an accessible functor on~$\Set$ is that it
possesses a final coalgebra. For an accessible $\Set$ functor~$F$ it
follows that the mapping $\fmorph{F}{X}{{\cdot}} : X \to \Omega$ is
well-defined, for each $F$-coalgebra~$(X,\alpha)$. Now, two elements
$x,y \in X$ are called \emph{behaviorally equivalent} for the
functor~$F$, notation $x \approx_F y$, if $\fmorph{F}{X}{x} =
\fmorph{F}{X}{y}$. In a way, if $x \approx_F y$ then $x$ and~$y$ can
be identified according to~$F$. Although $\fmorph{F}{X}{{\cdot}}$ is
an $F$-coalgebra homomorphism, i.e.\ $\omega \compose \mkern1mu
\fmorph{F}{X}{{\cdot}} = F( \fmorph{F}{X}{{\cdot}} ) \compose \alpha$,
one can argue that in fact the functor~$F$ determines which elements
of~$X$ are behavioral equivalent.


\section{State-to-function transition systems}
\label{sec:futs}

We start off with a formal definition of a state-to-function
transition system ($\FuTS$), we introduce the scheme for defining the
notion of bisimulation of a~$\FuTS$, and we relate $\FuTS$
bisimulation to behavioral equivalence for the type functor of
a~$\FuTS$.

\blankline

\begin{definition}
  \label{df-futs}
  A $\FuTS$ for a sequence of label sets $\calL_1 , \ldots ,
  \calL_n$ and a sequence of sequences of semirings $( \calR_{1,j}
  )_{j=1}^{m_1} , \ldots , ( \calR_{n,j} )_{j=1}^{m_n}$, for~$n, m_1,
  \ldots , m_n \gr 0$, is a pair~$\calX = ( X , \theta )$ of a set~$X$
  and a mapping 
  \begin{displaymath}
    \theta \, : \; X \to 
    \fsfn{X}{\calR_{1,1}, \ldots , \calR_{1,m_1} \mkern-3mu}^{\mkern1mu \calL_1}
    \, \times \mkern1mu \cdots \mkern1mu \times \,
    \fsfn{X}{\calR_{n,1}, \ldots , \calR_{n,m_n}
      \mkern-3mu}^{\mkern1mu \calL_n}
  \end{displaymath}
\end{definition}


\noindent
In the sequel, when $n \gr 1$, we usually represent the
mapping~$\theta$ as a tuple of mappings $< \theta_1 , \ldots ,
\theta_n >$ with $\theta_i : X \to \fsfn{X}{\calR_{i,1}, \ldots ,
  \calR_{i,m_i}}^{\mkern1mu \calL_i}$.
In the present paper we focus on the class of \emph{deterministic}
$\FuTS$, namely models where the transition relation is a function, as
in Definition~\ref{df-futs}. Note, this excludes by no means the
treatment of non-deterministic systems. Below and
in~\cite{DLLM13:cs,LMV12:accat} it has been shown that the class
of~$\FuTS$ given by the definition above is sufficiently rich to deal
with all the major stochastic process description languages and their
underlying semantic models.

In the next two sections we will discuss several examples of $\FuTS$
coming in various flavors. In particular we distinguish the cases of a
simple $\FuTS$, of a combined $\FuTS$, and of a nested $\FuTS$.
A~$\FuTS$ of the form $\calX = ( X , \theta )$ with $\theta : X \to
\fsfn{X}{\calR}^{\mkern1mu \calL}$, thus $n=1$, $m_1=1$, is called a
\emph{simple}~$\FuTS$. We will see that the familiar labeled
transition systems ($\LTS$) over an action set~$\calA$ are simple
$\FuTS$ for label set~$\calA$ and semiring~$\bools$, i.e.\ an~$\LTS$
with set of states~$S$ over~$\calA$ corresponds to a simple $\FuTS$
with transition function $\theta : S \to
\fsfn{S}{\bools}^{\calA}$. Likewise, discrete-time and continuous-time
Markov chains ($\DTMC$, $\CTMC$) are simple $\FuTS$ over a degenerate
label set~$\Delta$ and semiring~$\nnreals$. Putting $\Delta =
\singleton{\delta}$, a Markov chain with set of states~$S$ can be
identified with a simple $\FuTS$ with transition function $\theta : S
\to \fsfn{S}{\nnreals}^{\Delta}$. For a $\DTMC$ we will have
$\theta(s)(\delta)[S] = 1$.

A~$\FuTS$ of the form $\calX = ( X , \theta )$ with $\theta = <
\theta_1 , \ldots , \theta_n >$ and $\theta_i : X \to
\fsfn{X}{\calR_i}^{\mkern1mu \calL_i}$, for $i = 1 \ldots n$, thus
$m_1, \ldots, m_n = 1$, is called a
\emph{combined}~$\FuTS$. Interactive Markov chains ($\IMC$,
cf.~\cite{Her02:springer}) are prominent examples of
combined~$\FuTS$. An~$\IMC$ with states from~$S$ and action
set~$\calA$ can be seen as a combined $\FuTS$ for the label sets
$\calA$ and~$\Delta$, and semirings $\bools$ and~$\nnreals$ with a
pair of transition functions $< \theta_1 , \theta_2 >$. Here $\theta_1
: S \to \fsfn{S}{\bools}^\calA$ captures the interactive component of
the~$\IMC$ and $\theta_2 : S \to \fsfn{S}{\nnreals}^\Delta$ captures
the Markovian component.

Finally, a~$\FuTS$ of the form $\calX = ( X , \theta )$ with $\theta :
X \to \fsfn{X}{\calR_1 , \ldots , \calR_m}^{\mkern1mu \calL}$ is
called a \emph{nested}~$\FuTS$. Below we will argue that probabilistic
automata ($\PA$, cf.~\cite{SL95:njc,Hen12:facs}) over an action
set~$\calA$ are nested $\FuTS$ for the label set~$\calA$ and semirings
$\nnreals$ and~$\bools$. A~$\PA$ with states from~$S$ and actions
from~$\calA$ induces a two-level nested $\FuTS$ with transition
function $\theta : S \to \fsfn{S}{\nnreals,\bools}^\calA$, or more
explicitly $\theta : S \to \fsfn{ \, \fsfn{S}{\nnreals} \,
}{\bools}^\calA$. For continuations $\pi \in \fsfn{S}{\nnreals}$
involved, it will hold that $\pi[S] = 1$.

Most automata in the setting of quantitative process languages fall in
the three special types of~$\FuTS$ mentioned (simple, combined or
nested). An important semantic model not captured is that of Markov
automaton (MA, cf.~\cite{EHZ10:lics,EHZ10:concur,TKPS12:concur}). A Markov
automaton can be seen as a `general' $\FuTS$, i.e.\ a~$\FuTS$ that is
not of one of the distinguished types. More precisely, an~$\MA$ with
set of states~$S$ and action set~$\calA$ can be represented as a
$\FuTS$ for the label sets $\calA$ and~$\Delta$ and sequences of the
two semirings $\nnreals,\bools$ and of only one semiring~$\nnreals$,
thus having a pair of transition functions $< \theta_1 , \theta_2 > :
S \to \fsfn{ \, \fsfn{S}{\nnreals} \, }{\bools}^\calA \times
\fsfn{S}{\nnreals}^\Delta$. Here, $\theta_1$~represents the so-called
immediate transition relation, while $\theta_2$~represents the
so-called timed transition relation.

\blankline

\noindent
Next we define a notion of (strong) bisimulation~$\simeq_\calX$ for
a~$\FuTS$~$\calX = ( X , \theta \mkern1mu )$. By the definition below,
a bisimulation relation~$E$ is an equivalence relation on the set of
states~$X$. The relation~$E$ on~$X$ is then lifted to an equivalence
relation~$\calT(E)$ on~$\calT(X)$, invoking the so-called type~$\calT$
of the~$\FuTS$ (formally given in
Definition~\ref{df-futs-type}). For~$E$ to be a $\FuTS$ bisimulation
we require, that $E(x,y)$ for states $x$ and~$y$ implies $\calT(E)(
\theta(x) , \theta(y) )$.


\blankline

\begin{definition}
  \label{df-futs-bisimulation}
  Let $\calX = ( X , \theta )$ be a $\FuTS$ with transition function
  \begin{displaymath}
    \theta \, : \; X \to 
    \fsfn{X}{\calR_{1,1}, \ldots , \calR_{1,m_1} \mkern-3mu}^{\mkern1mu \calL_1}
    \, \times \mkern1mu \cdots \mkern1mu \times \,
    \fsfn{X}{\calR_{n,1}, \ldots , \calR_{n,m_n}
      \mkern-3mu}^{\mkern1mu \calL_n}
  \end{displaymath}
  An equivalence relation $E \subseteq {X \times X}$ is called a
  bisimulation for~$\calX$ if 
  \begin{displaymath}
    E(x,y)
    \implies
    \lift(E,\calR_{1,1}, \ldots , \calR_{1,m_1} \mkern-3mu)^{\mkern1mu \calL_1}
    \, \times \mkern1mu \cdots \mkern1mu \times \,
    \lift(E, \calR_{n,1}, \ldots , \calR_{n,m_n}
      \mkern-3mu)^{\mkern1mu \calL_n}
    \, ( \, \theta(x) ,\, \theta(y) \, )
  \end{displaymath}
  for all $x, y \in X$.  Two states $x,y \in X$ are $\calX$-bisimilar,
  written $x \simeq_\calX y$, if $E(x,y)$ for a bisimulation~$E$
  for~$\calX$.
\end{definition}

\blankline

\noindent
Recall, since $E$ is assumed to be an equivalence relation on~$X$, we
have that $\lift(E, \calR_{i,1} , \ldots , \calR_{i,m_i}
\mkern-3mu)^{\mkern1mu \calL_i}$ is an equivalence relation
on~$\fsfn{X}{\calR_{i,1} , \ldots , \calR_{i,m_i} \mkern-3mu
}^{\mkern1mu \calL_i}$, for all~$i$. Moreover, by
Definition~\ref{df-futs} we have $m_i \gr 0$ for all~$i$. Expanding
the definition of the outer $\lift( \mkern1mu {\cdot}
\mkern1mu,\calR_{i,m_i})$ for the relation involved, yields that if
two states are equivalent, i.e.\ $E(x,y)$, then evaluating
$\theta_i(x)(\myell)[C]$ and $\theta_i(y)(\myell)[C]$ amounts to the
same, for all labels~$\ell \in \calL_i$ and equivalence classes~$C$
of~$\lift(E,\calR_{i,1} , \ldots , \calR_{i,m_i-1})$, for all~$i$,
i.e.\ $\theta_i(x)(\myell)[C] = \theta_i(y)(\myell)[C]$, for $i = 1 ,
\ldots , n$.

For simple or combined $\FuTS$ the scheme is all straightforward since
$\lift(E,\calR_{i,1} , \ldots , \calR_{i,m_i-1})$ is just~$E$,
cf.~\cite{LMV12:accat}. However, when the codomain of the $\FuTS$
involves nested applications of the $\fsfn{\cdot}{\calR}$ operator,
$\fsfn{X}{\calR_{i,1} , \ldots , \calR_{i,m_i-1}}$ is of a higher
functional level than the set~$X$ itself. Therefore, we need to push
the relation~$E$ up, so to speak pushing it through the component
operators of the codomain.

In the sequel we will see several examples of quantitative automata
and their notion of bisimulation from the literature to coincide with
their $\FuTS$ representation and $\FuTS$ bisimulation of
Definition~\ref{df-futs-bisimulation}. However, the point is that
$\FuTS$ bisimulation can also be captured coalgebraically. In fact,
$\FuTS$ bisimulation and so-called behavioral
equivalence~\cite{Kur00:phd,Sta11:lmcs} are the same. Therefore, by
relating a quantitative execution model like a $\CTMC$, $\IMC$,
or~$\MA$ with a suitable~$\FuTS$, immediately provides coalgebraic
justification of the specific notion of bisimulation as the natural
(strong) process equivalence.

\blankline

\begin{definition}
  \label{df-futs-type}
  A $\FuTS$ $\calX = ( X , \theta )$ with transition function
  \begin{displaymath}
    \theta \, : \; X \to 
    \fsfn{X}{\calR_{1,1}, \ldots , \calR_{1,m_1} \mkern-3mu}^{\mkern1mu \calL_1}
    \, \times \mkern1mu \cdots \mkern1mu \times \,
    \fsfn{X}{\calR_{n,1}, \ldots , \calR_{n,m_n}
      \mkern-3mu}^{\mkern1mu \calL_n}
  \end{displaymath}
  is called a $\FuTS$ of type~$\calT$, for the $\Set$-functor~$\calT$
  given by
  \begin{displaymath}
    \calT = 
    \fsfn{ {\ldots} \fsfn{\cdot}{\calR_{1,1}} {\, \ldots \,}
    }{\calR_{1,m_1} \mkern-3mu}^{\mkern1mu \calL_1}  
    \, \times \mkern1mu \cdots \mkern1mu \times \,
    \fsfn{ {\ldots} \fsfn{\cdot}{\calR_{n,1}} {\, \ldots \,}
    }{\calR_{n,m_n} \mkern-3mu}^{\mkern1mu \calL_n} 
  \end{displaymath}
\end{definition}

\blankline

\noindent
Thus, if a $\FuTS$~$\calX$ is of type~$\calT$ for a functor~$\calT$,
then, in turn, $\calX$~is a coalgebra of~$\calT$. Note, $\calT$ is a
composition of $\Set$-functors: `finite support' functors
$\fsfn{\cdot}{\calR}$, exponentiation functors~$({\cdot})^{\mkern1mu
  \calL}$, and product functors $({\cdot}) \times ({\cdot})$. This
restricted form gives rise to the following result.

\blankline

\begin{theorem}
  \label{th-final-coalgebra}
  If a functor~$\calT$ on~$\Set$ is the type of a~$\FuTS$, then
  $\calT$ possesses a final coalgebra.
\end{theorem}

\begin{proof}
  Functors of the form $\fsfn{\cdot}{\calR}$ can be shown to be
  accessible using a standard argument,
  cf.~\cite{Kli09:mosses,BBBRS12}. Accessibility is preserved by
  products, exponentiation and composition. It follows that $\calT$
  itself is accessible, and hence has a final coalgebra,
  see~\cite{AP04:tcs}.
\end{proof}

\blankline

\noindent
Let $\calX = (X,\theta)$ be a $\FuTS$ of type functor~$\calT$ and let
$\bfOmega = ( \Omega , \omega )$ denote the final coalgebra
of~$\calT$. By finality of~$\bfOmega$ there exists a unique
$\calT$-homomorphism $\fmorph{\calX}{\calT}{{\cdot}} : X \to
\Omega$. Behavioral equivalence~$\approxT$ is then defined as $x
\approxT y$ iff $\fmorph{\calX}{\calT}{x} =
\fmorph{\calX}{\calT}{y}$. We have the following result  relating FuTS
bisimilarity~$\simeq_\calX$ to behavioral equivalence~$\approxT$ of
the type functor~$\calT$.

\blankline

\begin{theorem}[correspondence theorem]
  \label{th-futs-bisim-beh-equiv}
  Let $\calX = ( X , \theta )$ be a $\FuTS$ of type~$\calT$ for the
  functor~$\calT$ on~$\Set$. Then it holds that $x \simeq_\calX y$ iff
  $x \approxT y$, i.e.\ $\FuTS$ bisimulation and behavioral
  equivalence coincide.
  \qed
\end{theorem}

\blankline

\noindent
A restricted version of Theorem~\ref{th-futs-bisim-beh-equiv} was
given in~\cite{LMV12:accat,LMV13}. The present theorem generalizes the
result to also deal with nesting, a situation needed for the more
advanced quantitative automata discussed in the sequel. The proof of
the theorem is built on two lemmas.  To smooth the presentation, we
consider the lemmas only for non-product functors (i.e.\ choosing $ n=
1$ in Definition~\ref{df-futs-type}). The extension to product
functors is conceptually straightforward. 

Recall, for $f : X \rightarrow Y$, the functor application
$\fsfn{{\ldots} \fsfn{f}{\calR_{1}} {\, \ldots \,} }{\calR_{n}
  \mkern-3mu}^{\mkern1mu \calL} $ to~$f$ is a function from the set
$\fsfn{X, \calR_{1}, \ldots}{\calR_{n} \mkern-3mu}^{\mkern1mu \calL}$
to the set $ \fsfn{Y, \calR_{1}, \ldots}{\calR_{n}
  \mkern-3mu}^{\mkern1mu \calL}$ with $\fsfn{f}{\calR_1}(\varphi)(y) =
\tssum_{x \in f^{-1}(y)} \: \varphi(x)$ for $\varphi \in
\fsfn{X}{\calR_1}$ and $y\in Y$, and
\begin{displaymath}
  \fsfn{{\ldots} \fsfn{f}{\calR_{1}} {\, \ldots \,} }{\calR_{n}
  \mkern-3mu}(\Phi)(\ell)(\psi) = 
    \tssum_{\varphi \in \fsfn{{\ldots} \fsfn{f}{\calR_{1}} {\, \ldots
          \,} }{\calR_{n-1} \mkern-3mu}^{-1}(\psi)}\: \Phi(\ell)(\phi) 
\end{displaymath}
for $\Phi \in \fsfn{X, \calR_{1}, \ldots}{\calR_{n}
  \mkern-3mu}^{\mkern1mu \calL}$, $\ell \in \calL$, $\psi \in \fsfn{Y,
  \calR_{1}, \ldots}{\calR_{n-1} \mkern-3mu}$, and resulting value
in~$\calR_{n}$.


\blankline

\begin{lemma}
  \label{lm-nesting}
  Let $\calX = ( X , \theta )$ be a $\FuTS$ of type $\calT = \fsfn{
    {\ldots} \fsfn{\cdot}{\calR_{1}} {\, \ldots \,} }{\calR_{n}
    \mkern-3mu}^{\mkern1mu \calL} $. If $E$~is an equivalence relation
  on~$X$, then there exists a mapping $\theta_E \colon X/E \to \fsfn{
    X/E , \calR_1 , \ldots {} }{\calR_n}^{\mkern1mu \calL} $ such that
  $\calX_E = (X/E, \theta_E )$ is a coalgebra of the functor~$\calT$,
  and the canonical mapping $\varepsilon : X \to X/E$ is a
  $\calT$-homomorphism.
\end{lemma}

\begin{proof}[Sketch of proof.]
  Define a coalgebra structure~$\thetaE$ on~$X/E$ by putting
  \begin{displaymath}
    \thetaE( [x]_E )(\myell)(\zbar_{n-1})
    =
    \tssum_{z_{n-1} \in \fsfn{\varepsilon, \calR_1 , \ldots
      }{\calR_{n-1}}^{-1} ( \zbar_{n-1} )} \:
    \theta(x)(\myell)(z_{n-1})
  \end{displaymath}
  for~$x \in X$, $\ell \in \calL$, $\zbar_{n-1} \in \fsfn{X/E,
    \calR_1, \ldots }{\calR_{n-1}}$. 
  The property
  \begin{displaymath}
    \fsfn{\varepsilon, \calR_1, \ldots }{\calR_i}(z_i)
    =
    \fsfn{\varepsilon, \calR_1, \ldots }{\calR_i}(z'_i)
    \iff
    \lift(E, \calR_1, \ldots \calR_i)(z_i,z'_i)
  \end{displaymath}
  for $z_i, z'_i \in \fsfn{X, \calR_1, \ldots }{\calR_i}$, $i = 1,
  \ldots , n$, can be proved by induction on~$i$. Then, the property
  for~$n$ is the key ingredient to verify that $\fsfn{\varepsilon,
    \calR_1, \ldots }{\calR_n}^{\mkern1mu \calL} \compose \theta =
  \thetaE \compose \varepsilon$, i.e.\ $\varepsilon$~is a
  $\calT$-homomorphism.
\end{proof}

\blankline

\noindent
From the lemma it follows that the final mapping
$\fmorph{\calT}{\calX}{{\cdot}} : X \to \Omega$ factorizes
through~$\varepsilon$. Hence, if $\varepsilon(x) = \varepsilon(y)$,
then $x \approxT y$, proving half of
Theorem~\ref{th-futs-bisim-beh-equiv}. The reverse can be shown using
the following result.

\blankline

\begin{lemma}
  Let $\calX = ( X , \theta )$ be a $\FuTS$ of type $\calT = \fsfn{
    {\ldots} \fsfn{\cdot}{\calR_{1}} {\, \ldots \,} }{\calR_{n}
    \mkern-3mu}^{\mkern1mu \calL} $.  
  The relation $\approxT$ on~$X$
  is a bisimulation for the $\FuTS$~$\calX$.
\end{lemma}

\begin{proof}[Sketch of proof.]
  One first shows, by induction on~$i$,
  \begin{displaymath}
    y_i \in [x_i]_{\lift({\approxT}, \calR_1 , \ldots, \calR_i)} 
    \iff
    \fsfn{\fmorph{\calT}{\calX}{{\cdot}}, \calR_1 , \ldots
    }{\calR_i}(x_i) =
    \fsfn{\fmorph{\calT}{\calX}{{\cdot}}, \calR_1 , \ldots
    }{\calR_i}(y_i)
  \end{displaymath}
  for $x_i, y_i \in {\fsfn{X, \calR_1 , \ldots }{\calR_i}}$, and $i = 
  1, \ldots , n$. 
  Using the above property for~$n$, one next verifies
  \begin{displaymath}
    \omega( \mkern1mu \fmorph{\calT}{\calX}{x} \mkern1mu )(\myell)(w_{n-1}) 
    =
    \tssum_{z_{n-1} \in \fsfn{\fmorph{\calT}{\calX}{{\cdot}}, \calR_1
        , \ldots }{\calR_{n-1}}^{-1}(w_{n-1})} \:
    \theta(x)(\myell)(z_{n-1}) 
  \end{displaymath}
  for $x \in X$, $\ell \in \calL$, and $w_{n-1} \in \fsfn{\Omega,
    \calR_1, \ldots}{\calR_{n-1}}$.  From this the `transfer
  condition' of Definition~\ref{df-futs-bisimulation} for~$\approxT$
  follows.
\end{proof}

\blankline

\noindent
By the lemma, if $x \approxT y$, then there exists a bisimulation for
the $\FuTS$~$\calX$, viz.\ the bisimulation~$\approxT$, that relates
$x$ and~$y$. This proves the other direction of the correspondence
theorem.

In the next two sections, we proceed to incorporate the major
operational models used for the semantics of quantitative process
languages in the $\FuTS$
framework. Theorem~\ref{th-futs-bisim-beh-equiv} confirms that $\FuTS$
bisimulation is a proper notion of process equivalence. Thus, as a
consequence, a notion of process equivalence that coincides with
$\FuTS$ bisimulation also coincides with its coalgebraic counterpart.


\section{Quantitative transition systems as simple and 
  combined $\FuTS$}
\label{sec:simple}

In this section we interpret standard labeled transition systems,
continuous-time Markov chains and Hermann's interactive Markov chains
as $\FuTS$ and show that their usual notion of bisimulation coincides
with the notion of bisimulation of their associated~$\FuTS$.

\subsection{Labeled transition systems}
\label{sub:lts}

\noindent
It is straightforward to see that an~$\LTS$, over a set of
actions~$\calA$ and with a set of states~$S$, can be modeled as a
function $S \times \calA \to ( S \to \bools)$. However,
arbitrary~$\LTS$ do not fit in our set-up with finitely supported
functions. So, our modeling of~$\LTS$ here, similar as reported
elsewhere, e.g.\ \cite{KS08:fossacs,BBBRS12}, restricts to 
image-finite~$\LTS$.

\blankline

\begin{definition}
 Fix a set ~$\calA$  of actions.
  \begin{itemize}
  \item [(a)] An image-finite $\LTS$ over~$\calA$ is a pair $\calL = (
    S , {\trans{}{\calL}} )$ where $S$ is a set of states, and
    ${\trans{}{\calL}} \subseteq S \times \calA \times S$ is the
    transition relation such that, for all $s \in S$, $a \in \calA$,
    the set $\lc s' \mid s \trans{a}{\calL} s' \rc$ is finite.
  \item [(b)] An equivalence relation $R \subseteq S \times S$ is
    called a bisimulation equivalence for the $\LTS$~$\calL = ( S ,
    {\trans{}{\calL}} )$ if, for all $s,s',t \in S$, $a \in \calA$
    such that $R(s,t)$ and $s \trans{a}{\calL} s'$, there exists $t'
    \in S$ such that $t \trans{a}{\calL} t'$ and $R(s',t')$.
  \item [(c)] Two states $s,t \in S$ in an $\LTS$~$\calL = ( S ,
    {\trans{}{\calL}} )$ are called strongly bisimilar for
    $\LTS$~$\calL$, if there exists a bisimulation equivalence~$R$
    for~$\calL$ such that~$R(s,t)$. 
  \end{itemize}
\end{definition}

\blankline

\noindent
An image-finite $\LTS$ $\calL = ( S , {\trans{}{\calL}} )$
over~$\calA$ induces a simple $\FuTS$ $\calF(\calL) = ( S , \thetaL )$
over~$\calA$ and the semiring~$\bools$, if we define $\thetaL : S \to
\fsfn{S}{\bools}^{\mkern1mu \calA}$ by $\thetaL(s)(a)(s') \iff s
\trans{a}{\calL} s'$, for all $s, s' \in S$, $a \in \calA$. The next
theorem will not come as an surprise. See~\cite{KS08:fossacs,BBBRS12}
for example, for a proof that the respective notions of bisimulation
coincide for this specific case. In order to illustrate the general
pattern of such a proof for~$\FuTS$ we provide a proof here as well.

\blankline

\begin{theorem}
  \label{th-correspondence-lts}
  Let $\calL = ( S , {\trans{}{\calL}} )$ be an $\LTS$. Then it holds
  that $R$~is a bisimulation equivalence iff $R$~is a $\FuTS$
  bisimulation for~$\calF(\calL)$.
\end{theorem}

\begin{proof}
  The result follows almost directly from the definitions. We have,
  for $s,t \in S$,
  \begin{displaymath}
    \def\arraystrech{1.3}
    \begin{array}{@{}rcll@{}}
      \multicolumn{4}{l}{\lift(R,\bools)^{\mkern1mu \calA} (s,t)}
      \\ \quad & \iff &
      \forall \mkern1mu a \in \calA \mkern1mu
      \forall \mkern1mu C \in S/R \colon
      \tssum_{u \in C} \: \thetaL(s)(a)(u)
      =
      \tssum_{u \in C} \: \thetaL(t)(a)(u)
      & (\text{definition $\lift(R,\bools)^{\mkern1mu \calA}$})
      \\ & \iff &
      \forall \mkern1mu a \in \calA \mkern1mu
      \forall \mkern1mu C \in S/R \colon
      \exists \mkern1mu u \in C \colon s \trans{a}{\calL} u
      \Leftrightarrow
      \exists \mkern1mu u \in C \colon t \trans{a}{\calL} u
      & (\text{definition $\thetaL$})
      \\ & \iff &
      \forall \mkern1mu a \in \calA \mkern1mu
      \forall \mkern1mu s' \in S \colon
      s \trans{a}{\calL} s'
      \Rightarrow
      \exists \mkern1mu t' \in S \colon
      t \trans{a}{\calL} t' \land R(s',t')
      & (\text{$R$ is an equivalence relation})
    \end{array}
    \def\arraystrech{1.0}
  \end{displaymath}
  We use the logical equivalence from left to right in proving that a
  bisimulation for~$\calF(\calL)$ is a bisimulation equivalence, and the
  logical equivalence the other way around in proving that a
  bisimulation equivalence is a $\FuTS$ bisimulation.
\end{proof}

\blankline

\noindent
With appeal to the correspondence result,
Theorem~\ref{th-futs-bisim-beh-equiv}, we retrieve that strong
bisimulation and behavioral equivalence coincide.

\subsection{Continuous-time Markov chains}
\label{sub:ctmc}

As a first, basic example of a quantitative semantic model we consider
continuous-time Markov chains ($\CTMC$) and the notion of
lumpability. In its purest form, a~$\CTMC$ does not involve
actions. It can be viewed as connecting a state to a number of other
states while weighing the connection with a real number, viz.\ the
rate of the negative exponential distribution used to represent the
time associated with the transition. As for our treatment of~$\LTS$ we
need to restrict to image-finiteness here too, which amounts to finite
branching.

\blankline

\begin{definition}
  [cf.~\cite{BHKW06:entcs}] \mbox{}
  \begin{itemize}
  \item [(a)] A $\CTMC$ is a pair $\calC = ( S , {\trans{}{\calC}} )$
    where $S$ is a set of states, and ${\trans{}{\calC}} \subseteq S
    \times \nnreals \times S$ is the transition relation. Define
    $\bfR(s,s') = \tssum \lc \lambda \mid s \trans{\lambda}{\calC} s'
    \rc$ and $\bfR(s,C) = \tssum \lc \bfR(s,s') \mid s' \in C \rc$.
  \item [(b)] An equivalence relation $R \subseteq S \times S$ is
    called a lumping relation for the $\CTMC$~$\calC = ( S ,
    {\trans{}{\calC}} )$ if, for all $s,t \in S$, such that~$R(s,t)$
    it holds that $\bfR(s,C) = \bfR(t,C)$ for every equivalence
    class~$C$ of~$R$.
  \item [(c)] Two states $s,t \in S$ in a $\CTMC$~$\calC = ( S ,
    {\trans{}{\calC}} )$ are called lumping equivalent, if there
    exists a lumping relation~$R$ for~$\calC$ such that~$R(s,t)$.
  \end{itemize}
\end{definition}

\blankline

\noindent
A $\CTMC$ $\calC = ( S , {\trans{}{\calC}} )$ induces a simple $\FuTS$
$\calF(\calC) = ( S , \thetaC )$ over the label set~$\Delta =
\singleton{\delta}$ and the semiring~$\nnreals$, if we define $\thetaC
: S \to \fsfn{S}{\nnreals}^{\mkern1mu \Delta}$ by
\begin{displaymath}
  \thetaC(s)(\delta)(s') = 
  \tssum \lc \lambda \mid s \trans{\lambda}{\calC} s' \rc
\end{displaymath}
for $s, s' \in S$. Here, $\Delta$ is a dummy set to help $\CTMC$ fit
in the format of~$\FuTS$, cf.~\cite{BDL13:ic,DLLM13:cs}; conventionally, the
label~$\delta$ signifies delay.

\blankline

\begin{theorem}
  \label{th-correspondence-ctmc}
  Let $\calC = ( S , {\trans{}{\calC}} )$ be a $\CTMC$ and $R
  \subseteq S \times S$ an equivalence relation. Then it holds that
  $R$~is a lumping iff $R$~is a $\FuTS$ bisimulation
  for~$\calF(\calC)$.
\end{theorem}

\begin{proof}
  Also here the proof mainly consists of unfolding the various
  definitions. We have, for $s,t \in S$,
  \begin{displaymath}
    \def\arraystretch{1.3}
    \begin{array}{@{}rcll@{}}
      \multicolumn{4}{l}{\lift(R,\nnreals)^{\mkern1mu \Delta} (s,t)}
      \\ \quad & \iff &
      \forall \mkern1mu C \in S/R \colon
      \tssum_{u \in C} \: \thetaC(s)(\delta)(u)
      =
      \tssum_{u \in C} \: \thetaC(t)(\delta)(u)
      & (\text{definition $\lift(R,\nnreals)^{\mkern1mu \Delta}$})
      \\ & \iff &
      \forall \mkern1mu C \in S/R \colon
      \tssum \lc \lambda \mid s \trans{\lambda}{\calC} u ,\, u \in C \rc
      =
      \tssum \lc \mu \mid t \trans{\mu}{\calC} u ,\, u \in C
      \rc
      & (\text{definition $\thetaC$})
      \\ & \iff &
      \forall \mkern1mu C \in S/R \colon
      \bfR(s,C) = \bfR(t,C)
      & (\text{definition $\bfR$})
    \end{array}
    \def\arraystretch{1.0}
  \end{displaymath}
  Combing this logical equivalence with the respective definitions of
  lumping equivalence and bisimulation yields the result.
\end{proof}

\blankline

\noindent
For~$\CTMC$, from Theorem~\ref{th-futs-bisim-beh-equiv} and
elimination of the degenerated exponentiation with the singleton
set~$\Delta$, we obtain that lumping equivalence coincides with
behaviour equivalence of~$\fsfn{{\cdot}}{\nnreals}$.

\subsection{Interactive Markov chains}
\label{sub:imc}

\noindent
Interactive Markov chains ($\IMC$) were proposed
in~\cite{Her02:springer} as a reconciliation of $\LTS$
and~$\CTMC$. Because of this two-dimensionality $\IMC$ constitute a
prime example of a combined~$\FuTS$. The definition of an~$\IMC$ below
is taken from~\cite{HK10:fmco}.

\blankline

\pagebreak[3]

\begin{definition}
 Fix a set ~$\calA$  of actions.
  \begin{itemize}
  \item [(a)] An interactive Markov chain ($\IMC$) over~$\calA$ is a
    triple $\calI = ( S , {\trans{}{\calI}} , {\mtrans{}{\calI}} )$
    where $S$ is a set of states, ${\trans{}{\calI}} \subseteq S
    \times \calA \times S$ is the interactive transition relation, and
    ${\mtrans{}{\calI}} \subseteq S \times \nnreals \times S$ is the
    Markovian transition relation.
  \item [(b)] For states~$s,s' \in S$, action~$a \in \calA$, and a
    subset of states~$C \subseteq S$, define $\bfT(s,a,C)
    \Leftrightarrow s \trans{a}{\calI} \sbar$ for some state~$\sbar
    \in C$. Moreover, define $\bfR(s,s') = \tssum \lc \lambda \mid s
    \mtrans{\lambda}{\calI} s' \rc$ and $\bfR(s,C) = \tssum \lc
    \bfR(s,s') \mid s' \in C \rc$.
  \item [(c)] An equivalence relation $R \subseteq S \times S$ is
    called a bisimulation relation for the $\IMC$~$\calI = ( S ,
    {\trans{}{\calI}} , {\mtrans{}{\calI}} )$ if for all states~$s,t
    \in S$ and every equivalence class~$C$ of~$R$ the following holds:
    \begin{itemize}
    \item [(i)] $\bfT(s,a,C) = \bfT(t,a,C)$, for all $a \in \calA$;
    \item [(ii)] $\bfR(s,C) = \bfR(t,C)$.
    \end{itemize}
  \item [(c)] Two states $s,t \in S$ in an $\IMC$~$\calI = ( S ,
    {\trans{}{\calI}} , {\mtrans{}{\calI}} )$ are called bisimilar, if
    there exists a bisimulation relation~$R$ for~$\calI$
    with~$R(s,t)$. 
  \end{itemize}
\end{definition}

\blankline

\noindent
As for $\CTMC$ we use the symbol~$\delta$ to denote delay, and put
$\Delta = \singleton{\delta}$. An $\IMC$ $\calI = ( S ,
{\trans{}{\calI}} , {\mtrans{}{\calI}} )$ over~$\calA$ induces a
combined $\FuTS$ $\calF(\calI) = ( S , \thetaI )$, where $\thetaI = <
\thetaII , \thetaMI >$, over the label sets~$\calA$ and~$\Delta$ and
the semirings $\bools$ and~$\nnreals$. We define $\thetaII : S \to
\fsfn{S}{\bools}^\calA$ and $\thetaMI : S \to
\fsfn{S}{\nnreals}^\Delta$ by
\begin{displaymath}
  \thetaII(s)(a)(s') \Leftrightarrow s \trans{a}{\calI} s'
  \quad \text{and} \quad
  \thetaMI(s)(\delta)(s') = \bfR(s,s')
\end{displaymath}
for all $s, s' \in S$, $a \in \calA$. Thus, the transition
function~$\thetaII$ is similar to the transition function~$\thetaL$ of
an~$\LTS$. The transition function~$\thetaMI$ is similar to the
transition function~$\thetaC$ of a~$\CTMC$.

\blankline 

\noindent
The transition relation of a $\FuTS$ for an~$\IMC$ is the
superposition of those of an~$\LTS$ and a~$\CTMC$. Therefore, the
proof of a correspondence result of standard bisimulation and $\FuTS$
bisimulation for an~$\IMC$ combines the observations made in the
proofs of Theorems \ref{th-correspondence-lts}
and~\ref{th-correspondence-ctmc}. 

\blankline

\begin{theorem}
  Let $\calI = ( S , {\trans{}{\calI}} , {\mtrans{}{\calI}} )$ be
  an~$\IMC$ and $R \subseteq S \times S$ an equivalence relation. Then
  it holds that $R$~is a bisimulation for the $\IMC$~$\calI$ iff
  $R$~is a bisimulation for the $\FuTS$~$\calF(\calI)$.
\end{theorem}

\begin{proof}
  For an equivalence relation~$R$ and states $s,t \in S$ such
  that~$R(s,t)$ we have the following logical equivalence:
  \begin{displaymath}
    \def\arraystretch{1.3}
    \begin{array}{rclr}
      \multicolumn{4}{l}{\lift(R,\bools)^{\mkern1mu \calA} \times
        \lift(R,\nnreals)^{\mkern1mu \Delta} ( \thetaI(s) ,
        \thetaI(t) )}
      \\ \  & \Leftrightarrow &
      \lift(R,\bools)^{\mkern1mu \calA} \times
        \lift(R,\nnreals)^{\mkern1mu \Delta}( < \thetaII(s) ,
        \thetaMI(s) > , < \thetaII(t) , \thetaMI(t)> )
      & (\text{since $\thetaI = < \thetaII , \thetaMI >$})
      \\ & \Leftrightarrow &
      \lift(R,\bools)^{\mkern1mu \calA} ( \thetaII(s) , \thetaII(t) )
      \; \land \;
      \lift(R,\nnreals)^{\mkern1mu \Delta} ( \thetaMI(s) ,
      \thetaMI(t) )
      & (\text{definition of relational product})      
      \\ & \Leftrightarrow &
      \multicolumn{2}{l}{%
      \forall \mkern1mu a \in \calA
      \forall \mkern1mu C \in S/R \colon 
      \bigvee \lc \thetaII(s)(a)(u) \mid u \in C \rc
      =
      \bigvee \lc \thetaII(t)(a)(u) \mid u \in C \rc \; \land {}
      } 
      \\ & &
      \multicolumn{2}{l}{%
      \forall \mkern1mu C \in S/R \colon
      \tssum \lc \thetaMI(s)(\delta)(u) \mid u \in C \rc
      =
      \tssum \lc \thetaMI(t)(\delta)(u) \mid u \in C \rc
      } 
      \\
      \multicolumn{4}{r}{(\text{definition of
          $\lift(R,\bools)^{\mkern1mu \calA}$ and
          $\lift(R,\nnreals)^{\mkern1mu \Delta}$})}   
      \\ & \Leftrightarrow &
      \multicolumn{2}{l}{%
      \forall \mkern1mu a \in \calA
      \forall \mkern1mu C \in S/R \colon 
      \exists \mkern1mu u \in C : \thetaII(s)(a)(u) 
      \Leftrightarrow
      \exists \mkern1mu u \in C : \thetaII(t)(a)(u)  
      \; \land {}
      } 
      \\ & & 
      \forall \mkern1mu C \in S/R \colon 
      \tssum \lc \bfR(s,u) \mid u \in C \rc
      =
      \tssum \lc \bfR(t,u) \mid u \in C \rc
      & (\text{definition of sum on $\bools$ and~$\nnreals$})  
      \\ & \Leftrightarrow &
      \multicolumn{2}{l}{%
      \forall \mkern1mu a \in \calA
      \forall \mkern1mu C \in S/R \colon 
      \bfT(s,a,C) = \bfT(t,a,C)
      \; \land \;
      \forall \mkern1mu C \in S/R \colon 
      \bfR(s,C) = \bfR(t,C) 
      } 
      \\ & & 
      & (\text{definition of $\bfT$ and~$\bfR$})
    \end{array}
    \def\arraystretch{1.0}
  \end{displaymath}
  Thus, if $R$ is a $\FuTS$ bisimulation for~$\calI$, then $R(s,t)$
  implies $\lift(R,\bools)^{\mkern1mu \calA} \times
  \lift(R,\nnreals)^{\mkern1mu \Delta} ( \thetaI(s) , \thetaI(t)
  )$. Hence, for all $C \in S/R$, we have $\bfT(s,a,C) = \bfT(t,a,C)$
  for all $a \in \calA$, and $\bfR(s,C) = \bfR(t,C)$. So, $R$~is an
  $\IMC$ bisimulation. Reversely, if $R$ is an $\IMC$ bisimulation,
  then $R(s,t)$ implies for all $C \in S/R$, we have $\bfT(s,a,C) =
  \bfT(t,a,C)$ for all $a \in \calA$, and $\bfR(s,C) =
  \bfR(t,C)$. Thus, $\lift(R,\bools)^{\mkern1mu \calA} \times
  \lift(R,\nnreals)^{\mkern1mu \Delta} ( \thetaI(s) , \thetaI(t)
  )$. So, $R$~is a $\FuTS$ bisimulation.
\end{proof}

\blankline

\noindent
For~$\IMC$ we have another proof of concrete bisimilarity being equal
to behavioral equivalence, a result also presented
in~\cite{LMV12:accat}. However, here we see better how the
bisimulation scheme guides the correspondence result for standard
bisimulation for~$\IMC$, on the one hand, and $\FuTS$ bisimulation, on
the other hand.


\section{Quantitative automata as nested and general $\FuTS$}
\label{sec:nested}

In this section we show that $\FuTS$ and their associated notion of
bisimulation suit probabilistic automata as well as Markov
automata. For the latter fact to prove we need the full generality of
Theorem~\ref{th-futs-bisim-beh-equiv}.

\subsection{Probabilistic automata}
\label{sub:pts}

\noindent
As next quantitative semantic model we consider probabilistic automata
($\PA$) originating from~\cite{SL95:njc}, and the associated notion of
strong Segala bisimulation. We follow the set-up presented
in~\cite{Hen12:facs}.

\blankline

\begin{definition}
  \begin{itemize}
    Fix a set of actions~$\calA$. 
  \item [(a)] A $\PA$ over~$\calA$ is a pair $\calP = ( S ,
    {\trans{}{\calP}} )$ where $S$ is a set of states, and
    ${\trans{}{\calP}} \subseteq S \times \calA \times \Distr(S)$ is an
    image-finite transition relation, i.e.\ the set $\lc \pi \mid s
    \trans{a}{\calP} \pi \rc$ is finite, for any state~$s \in S$, and
    action~$a \in \calA$.
  \item [(b)] An equivalence relation $R \subseteq S \times S$ is
    called a bisimulation relation for the $\PA$~$\calP = ( S ,
    {\trans{}{\calP}} )$ if, for all $s,t \in S$, $a \in \calA$, $\pi
    \in \Distr(S)$ such that $R(s,t)$ and $s \trans{a}{\calP} \pi$,
    there exists $\varrho \in \Distr(S)$ such that $t \trans{a}{\calP}
    \varrho$ and $\pi[C] = \varrho[C] \mkern1mu$, for every
    equivalence class~$C$ of~$R$.
  \item [(c)] Two states $s,t \in S$ in a $\PA$~$\calP = ( S ,
    {\trans{}{\calP}} )$ are called probabilistically bisimilar, if
    there exists a bisimulation relation~$R$ for~$\calP$ such
    that~$R(s,t)$.
  \end{itemize}
\end{definition}

\blankline

\noindent
A $\PA$ $\calP = ( S , {\trans{}{\calP}} )$ over~$\calA$ induces a
nested $\FuTS$ $\calF(\calP) = ( S , \thetaP )$ with label set~$\calA$
and semirings $\nnreals$ and~$\bools$. We define $\thetaP : S \to
\fsfn{ \mkern1mu \fsfn{S}{\nnreals} \mkern1mu }{\bools}^{\mkern1mu
    \calA}$ by
\begin{displaymath}
  \thetaP(s)(a)(\varphi)
  \iff
  s \trans{a}{\calP} \varphi
\end{displaymath}
for all $s \in S$, $a \in \calA$, $\varphi : S \to \nnreals$.  Note
the nesting of $\fsfn{{\cdot}}{\nnreals}$
and~$\fsfn{{\cdot}}{\bools}$. Also note that, if
$\thetaP(s)(a)(\varphi) = \TRUE$ then $\varphi$ is in fact a
probability distribution, since the probabilistic transition relation
connects states to probability distributions only. Finally note, if
$\thetaP(s)(a)(\varphi) = \FALSE$ for all~$\varphi$, then $\calP$
admits no $a$-transition for~$s$.

The proof of the correspondence result for~$\PA$ is along the same
lines as we have seen previously. However, now the equivalence
relation~$R$ needs to be lifted twice: to the level of
$\fsfn{S}{\nnreals}$ first, and to the level of $\fsfn{ \,
  \fsfn{S}{\nnreals} \, }{\bools}$ next.

\blankline

\begin{theorem}
  \label{th-correspondence-pa}
  Let $\calP = ( S , {\trans{}{\calP}} )$ be a $\PA$ and $R \subseteq 
  S \times S$ an equivalence relation. Then it holds that $R$~is a
  probabilistic bisimulation iff $R$~is a $\FuTS$ bisimulation
  for~$\calF(\calP)$.
\end{theorem}

\begin{proof}
  For an equivalence relation~$R \subseteq S \times S$ and states $s,t
  \in S$ such that~$R(s,t)$ we have the following:
  \begin{displaymath}
    \def\arraystretch{1.3}
    \begin{array}{rclr}
      \multicolumn{4}{l}{%
        \lift( \, \lift(R,\nnreals) , \bools)^{\mkern1mu \calA}
        \mkern1mu ( \mkern1mu \thetaP(s) , \thetaP(t) \mkern1mu )
      } 
      \\ \quad & \Leftrightarrow & 
      \forall \mkern1mu a \in \calA \mkern1mu
      \forall \mkern1mu \Gamma \in
      \fsfn{S}{\nnreals}/\lift(R,\nnreals) \colon
      \thetaP(s)(a)(\Gamma) 
      \Leftrightarrow
      \thetaP(t)(a)(\Gamma) 
      & \qquad \\ \multicolumn{4}{r}{%
        (\text{definition $\lift( \lift(R,\nnreals),\bools)^{\mkern1mu \calA}$}) 
      } 
      \\ & \Leftrightarrow &
      \forall \mkern1mu a \in \calA \mkern1mu
      \forall \mkern1mu \Gamma \in
      \fsfn{S}{\nnreals}/\lift(R,\nnreals) \colon \:
      \\ & & 
      \multicolumn{2}{r}{%
      \qquad
      \exists \mkern1mu \varphi \in \Gamma \colon
      \thetaP(s)(\varphi) = \TRUE
      \; \Leftrightarrow \;
      \exists \mkern1mu \psi \in \Gamma \colon
      \thetaP(t)(\psi) = \TRUE
      \qquad \qquad
      (\text{ring structure $\bools$}) 
      } 
      \\ & \Leftrightarrow &
      \multicolumn{2}{l}{%
      \forall \mkern1mu a \in \calA \mkern1mu
      \forall \mkern1mu \Gamma \in
      \fsfn{S}{\nnreals}/\lift(R,\nnreals) \colon \:
      \exists \mkern1mu \varphi \in \Gamma \colon 
      s \trans{a}{\calP} \varphi
      \; \Leftrightarrow \;
      \exists \mkern1mu \psi \in \Gamma \colon 
      t \trans{a}{\calP} \psi
      } 
      \\ & \Leftrightarrow &
      \multicolumn{2}{l}{%
      \forall \mkern1mu a \in \calA \mkern1mu
      \forall \mkern1mu \pi \in \Distr(S) \colon
      s \trans{a}{\calP} \pi
      \Rightarrow
      \exists \mkern1mu \varrho \in \Distr(s) \colon
      t \trans{a}{\calP} \varrho \land
      \forall C \in S/R \colon \pi[C] = \varrho[C]
      } 
      \\ \multicolumn{4}{r}{%
        (\text{symmetry of~$R$, $s \trans{a}{\calP} \varphi$ implies
          $\varphi \in \Distr(S)$}) 
      } 
    \end{array}
    \def\arraystretch{1.0}
  \end{displaymath}
  As before, combination of the above equivalence and the
  respective definitions of $\PA$~bisimulation for~$\calP$ and $\FuTS$
  bisimulation for~$\calF(\calP)$ yields the result.
\end{proof}


\blankline

\noindent
Thus for~$\PA$ we obtain, via Theorem~\ref{th-futs-bisim-beh-equiv},
coalgebraic underpinning using $\FuTS$ as an intermediate model.

\subsection{Markov automata}
\label{sub:ma}

Markov automata ($\MA$) are a relative recent example of a
quantitative semantical model~\cite{EHZ10:lics,EHZ10:concur,TKPS12:concur}. It brings together non-deterministic
and probabilistic choice, and stochastic delay. 

\blankline

\begin{definition}
  Fix a set of actions~$\calA$. 
  \begin{itemize}
  \item [(a)] A Markov automaton ($\MA$) over~$\calA$ is a triple
    $\calM = ( S , {\trans{}{\calM}} , {\mtrans{}{\calM}} )$ where $S$
    is a set of states, ${\trans{}{\calM}} \subseteq S \times
    \calA \times \Distr(S)$ is the immediate transition relation, and
    ${\mtrans{}{\calM}} \subseteq S \times \nnreals \times S$ is the
    timed transition relation.
  \item [(b)] 
    For states~$s,s' \in S$, action~$a \in \calA$, and a set of
    distributions~$\Gamma \subseteq \Distr(S)$, define
    $\bfT(s,a,\Gamma \mkern1mu) \Leftrightarrow s \trans{a}{\calM}
    \pi$ for some distribution~$\pi \in \Gamma$. Moreover, define
    $\bfR(s,s') = \tssum \lc \lambda \mid s \mtrans{\lambda}{\calM} s'
    \rc$, and, for a set of states~$C$, $\bfR(s,C) = \tssum \lc
    \bfR(s,s') \mid s' \in C \rc$.
  \item [(c)] An equivalence relation $R \subseteq S \times S$ is
    called a bisimulation relation for the $\MA$~$\calM = ( S ,
    {\trans{}{\calM}} , {\mtrans{}{\calM}} )$ if, for all states~$s,t
    \in S$ 
    the following holds:
    \begin{itemize}
    \item [(i)] $\bfT(s,a,\Gamma \mkern1mu ) =
      \bfT(t,a,\Gamma \mkern1mu )$,  
    for all $a \in \calA$, and $\Gamma \in
    \Distr(S)/\lift(R, \nnreals)_{\Distr(S)}$; 
    \item [(ii)] $\bfR(s,C) = \bfR(t,C)$, for all $C\in S/R$.
    \end{itemize}
    where $\lift(R, \nnreals)_{\Distr(S)}$ is the lifting of~$R$
    to~$\Distr(S)$; in the sequel we will often omit the subscript
    $_{\Distr(S)}$ for the sake of readability.
  \item [(d)] Two states $s,t \in S$ in an $\MA$~$\calM = ( S ,
    {\trans{}{\calM}} , {\mtrans{}{\calM}} )$ are called bisimilar, if
    there exists a bisimulation relation~$R$ for~$\calM$
    with~$R(s,t)$. 
  \end{itemize}
\end{definition}

\blankline

\noindent
In part~(c), the definition of a bisimulation relation for an~$\MA$,
we encounter the use of two kinds of equivalence relations. For the
timed behaviour the comparison is at the level of states involving
equivalence classes from~$S/R$. For the immediate behaviour, though,
the comparison is at the level of distributions of states involving
equivalence classes from~$\Distr(S)/\lift(R, \nnreals)_{\Distr(S)}$,
where $\lift(R, \nnreals)_{\Distr(S)}$ denotes the lifting of~$R$ with
respect to~$\Distr(S)$, see Section~\ref{sec:preliminaries}. Note,
for~$\PA$, because of the direct definition used, the application of
the lifting operator is left implicit.

Again, distinguish the symbol~$\delta$ to denote delay, and put
$\Delta = \singleton{\delta}$. An~$\MA$ $\calM = ( S ,
{\trans{}{\calM}} , {\mtrans{}{\calM}} )$ over~$\calA$ induces a
general $\FuTS$ $\calF(\calM) = ( S , \thetaM )$, with $\thetaM = <
\thetaIM , \thetaMM >$, over the label sets~$\calA$ and~$\Delta$ and
the sequences of semirings $\nnreals$, $\bools$ and~$\nnreals$, if we
define $\thetaIM : S \to \fsfn{ \, \fsfn{S}{\nnreals} \,
}{\bools}^{\mkern1mu \calA}$ and $\thetaMM : S \to
\fsfn{S}{\nnreals}^{\mkern1mu \Delta}$ by
\begin{displaymath}
  \thetaIM(s)(a)(\varphi) \Leftrightarrow s \trans{a}{\calM} \varphi 
  \quad \text{and} \quad
  \thetaMM(s)(\delta)(s') = \bfR(s,s')
\end{displaymath}
for $s, s' \in S$, $a \in \calA$, $\varphi \in \fsfn{S}{\nnreals}$.
Note, the $\FuTS$~$\calF(\calM)$ is a combination of a nested~$\FuTS$
representing the immediate behaviour of~$\calM$ and a simple~$\FuTS$
representing the timed behaviour of~$\calM$.

\blankline

\begin{theorem}
  Let $\calM = ( S , {\trans{}{\calM}} , {\mtrans{}{\calM}} )$ be
  a~$\MA$ and $R \subseteq S \times S$ an equivalence relation. Then
  it holds that $R$~is an $\MA$ bisimulation for~$\calM$ iff $R$~is a
  $\FuTS$ bisimulation for~$\calF(\calM)$.
\end{theorem}

\begin{proof}
  For a binary relation~$R \subseteq S \times S$ and states $s,t \in
  S$ such that~$R(s,t)$ we have the following logical equivalence:
  \begin{displaymath}
    \def\arraystretch{1.3}
    \begin{array}{rclr}
      \multicolumn{4}{l}{%
        \bigl( \mkern1mu
        \lift( \, \lift(R, \nnreals), \bools)^{\mkern1mu \calA}
        \: \times \:
        \lift(R, \nnreals)^{\mkern1mu \Delta}
        \mkern1mu \bigr)
        \mkern1mu ( \mkern1mu \thetaM(s) , \thetaM(t) \mkern1mu )
      } 
      \\ \quad & \Leftrightarrow & 
      \lift(\lift(R,\nnreals), \bools)^{\mkern1mu \calA}
      \mkern1mu ( \mkern1mu \thetaIM(s) , \thetaIM(t) \mkern1mu )
      \land
      \lift(R,\nnreals)^{\mkern1mu \Delta}
      \mkern1mu ( \mkern1mu \thetaMM(s) , \thetaMM(t) \mkern1mu )
      & \qquad \\ \multicolumn{4}{r}{%
        (\text{$\thetaM = < \thetaIM , \thetaMM >$, definition
          relational product}) 
      } 
      \\ & \Leftrightarrow &
      \forall \mkern1mu a \in \calA \mkern1mu
      \forall \mkern1mu \pi \in \Distr(S) \colon
      \\ & & \qquad
      s \trans{a}{\calP} \pi
      \Rightarrow
      \exists \mkern1mu \varrho \in \Distr(s) \colon
      t \trans{a}{\calP} \varrho \land
      \forall C \in S/R \colon \pi[C] = \varrho[C]
      \; \land \; 
      \\ & &
      \forall \mkern1mu C \in S/R \colon
      \tssum \lc \lambda \mid s \trans{\lambda}{\calC} u ,\, u \in C \rc
      =
      \tssum \lc \mu \mid t \trans{\mu}{\calC} u ,\, u \in C
      \rc
      & \\ \multicolumn{4}{r}{%
        (\text{see the proofs of Theorems \ref{th-correspondence-pa}
          and~\ref{th-correspondence-ctmc}})  
      } 
      \\ & \Leftrightarrow &
      \forall \mkern1mu a \in \calA \mkern1mu
      \forall \Gamma \in \Distr(S)/\lift(R)
      \colon 
      \bfT(s,a,\Gamma \mkern1mu ) = \bfT(t,a,\Gamma \mkern1mu ) \; \land
      \\ & & 
      \forall C \in S/R \colon 
      \bfR(s,C) = \bfR(t,C)
      & \\ \multicolumn{4}{r}{%
        (\text{symmetry~$R$, definition $\bfT$ and~$\bfR$})  
      } 
    \end{array}
    \def\arraystretch{1.0}
  \end{displaymath}
  As previously, combination of the above equivalence and the
  respective definitions of $\MA$~bisimulation for~$\calM$ and $\FuTS$
  bisimulation for~$\calF(\calM)$ yields the result.
\end{proof}

\blankline

\noindent
Markov automata feature, so to speak, simultaneously non-deterministic
branching followed by probabilistic branching on the one hand,
vs.\ stochastic branching on the other hand. This is reflected by the
corresponding $\FuTS$ being of type $\fsfn{ \, \fsfn{ {\cdot}
  }{\nnreals} \, }{\bools}^{\mkern1mu \calA} \: \times \: \fsfn{
  {\cdot} }{\nnreals}^{\mkern1mu \Delta}$. Again, by the
correspondence theorem, behavioral equivalence of the type functor
captures exactly the concrete notion of bisimulation.


\section{Concluding remarks}
\label{sec:conclusions}

We contributed to the work of providing uniform techniques for
modeling quantitative process languages. The concept of a~$\FuTS$,
originating from~\cite{DLLM09:icalp}, provides a compact way to assign
quantities to states, or other entities, by means of continuations,
leading to clean and concise descriptions of~$\SPC$,
cf.~\cite{DLLM13:cs}.

In the current paper, we extended the results we presented
in~\cite{LMV12:accat}. Here, we associate a type to a~$\FuTS$, and we
propose a systematic way, directed by the type, of lifting an
equivalence relation from the level of states to the level of the
continuations involved. The scheme allows types with products and
arbitrary nesting. The induced notion of $\FuTS$ bisimulation
coincides with behavioral equivalence, a coalgebraic notion of
identification associated with the $\Set$ functor implied by the type.
In particular, the correspondence result proved in~\cite{LMV12:accat}
now extends to nested $\FuTS$ and related functors. Various forms of
quantitative transition systems are shown to be amenable to a
representation as a $\FuTS$ including $\IMC$, $\PA$
and~$\MA$. Moreover, the concrete notion of bisimulation of these
quantitative transition systems is shown to coincide with the notion
of $\FuTS$ bisimulation as given by our scheme, and hence with
behavioral equivalence.

The main restriction of $\FuTS$ is its being based on finitely
supported functions. As a consequence, image-finiteness is inherent to
our treatment. Still the scope of application is broad. Apart from the
examples discussed here and in our earlier
work~\cite{LMV12:accat,LMV13} we have been able to model discrete
real-time with $\FuTS$ as well~\cite{LMV15:mscs}. However, replacing
the construct $\fsfn{ \mkern1mu {\cdot} \mkern2mu }{\calR}$ by the
Giry monad of measurable functions~\cite{Pan09:icp} may be an option:
on the one hand, finiteness is traded for a restricted form of
infinite, on the other hand, summation is exchanged for
integration. However, in the continuous setting, e.g.\ for hybrid
systems, broadly accepted semantical models that parallel $\LTS$ seem
to be missing.


In this paper we focused on strong notions of bisimilarity.  Weak
notions prove difficult to handle, e.g.\ for probabilistic weak
bisimulation~\cite{EHZ10:concur} or for probabilistic branching
bisimulation~\cite{AGT12:tcs}, see also~\cite{SVW09:sacs}. 
It would be interesting to investigate whether the more abstract view
based on $\FuTS$ and a categorical approach, in the line
of~\cite{MiP13,BMP14} could be of help.

\emph{Acknowledgments} The authors are grateful to Rocco De~Nicola,
Fabio Gadducci, Daniel Gebler, Michele Loreti, Jan Rutten, and Ana
Sokolova for fruitful discussions on the subject and useful
suggestions. DL and MM acknowledge support by EU Project n.~257414
{\em Autonomic Service-Components Ensembles} (ASCENS) and EU Project
n.~600708 \emph{A Quantitative Approach to Management and Design of
  Collective and Adaptive Behaviours} (QUANTICOL).

\bibliographystyle{eptcs}
\bibliography{qapl15}

\end{document}